\documentclass[12pt]{article}
\usepackage{amssymb}
\usepackage{amsfonts}

\newtheorem{theorem}{Theorem}

\newtheorem{conjecture}[theorem]{Conjecture}
\newtheorem{corollary}[theorem]{Corollary}

\newtheorem{remark}[theorem]{Remark}

\newenvironment{proof}[1][Proof]{\noindent\textbf{#1.} }{\ \rule{0.5em}{0.5em}}
\input{tcilatex}
\begin{document}

\title{A New Approach Towards the Golomb-Welch Conjecture}
\author{Peter Horak$^{\ast }$,~ Otokar Gro\v{s}ek$^{\dotplus }$~ \\
$^{\ast }$University of Washington, Tacoma\\
$^{\dotplus }$Slovak University of Technology, Slovakia\medskip \\
This paper is dedicated to the 75-th birthday of Toma\v{s} Janovic}
\maketitle

\begin{abstract}
\noindent The Golomb-Welch conjecture deals with the existence of perfect $e$%
-error correcting Lee codes of word length $n,$ $PL(n,e)$ codes. Although
there are many papers on the topic, the conjecture is still far from being
solved. In this paper we initiate the study of an invariant connected to
abelian groups that enables us to reformulate the conjecture, and then to
prove the non-existence of linear $PL(n,2)$ codes for $n\leq 12$. Using this
new approach we also construct the first quasi-perfect Lee codes for
dimension $n=3,$ and show that, for fixed $n$, there are only finitely many
such codes over $%
\mathbb{Z}
.$
\end{abstract}

\noindent It turns out that the Lee metric is more suitable for some
applications than the most frequently used Hamming metric. The Lee metric
has been used for the first time in \cite{Lee} and \cite{Ulrich} when
dealing with transmission of signals over noisy channels. Since then several
types of codes in Lee metric have been studied. For example, the perfect
error-correcting Lee codes introduced in \cite{GW}, the negacyclic codes
introduced by Berlekamp \cite{B}, see also \cite{A1}, \cite{Bose}, \cite{ADH}%
, \cite{E}, and \cite{HB} for other types and results on Lee codes. \medskip

\noindent In this paper we focus on perfect and quasi-perfect
error-correcting Lee codes. Except for practical applications, the
Golomb-Welch conjecture \cite{GW} on the existence of perfect Lee codes has
been the main motive power behind the research in the area for more than 40
years. Although there are many papers on the topic, the conjecture is far
from being solved. In these papers the authors use various methods when
attacking the conjecture. However, each of these methods has its limitation
and will not enable one to settle the conjecture completely. More detailed
account on the methods used will be given in Section 3. Thus, in this paper
we initiate the study of a new approach for tackling the conjecture. We have
looked for a setting for the Golomb-Welch conjecture, also the G-W
conjecture, in the area with a well developed theory containing many deep
results. We have chosen an approach based on a new invariant related to
homomorphisms of abelian groups. We will show how this invariant relates to
linear $PL(n,e)$ codes. Using our approach we prove the non-existence of
linear perfect $2$-error correcting codes for $n=7,...,11.$ Proving the G-W
conjecture for linear codes would constitute a big progress. To complete the
proof of the Golomb-Welch conjecture it would be needed to show that if
there is no linear perfect Lee code then there is no perfect Lee code. In
other words, if there is no lattice tiling of $%
\mathbb{Z}
^{n}$ by Lee spheres of radius $e,$ then there is no tiling of $%
\mathbb{Z}
^{n}$ by such Lee spheres. This means to answer in the affirmative a very
special case of the second part of the Hilbert's 18th problem.\ For more
information on the problem we refer the reader to \cite{Keaton} and \cite%
{Milnor}.\medskip

\noindent Although the G-W conjecture has not been solved yet it is widely
believed that it is true. Therefore, instead of searching for perfect Lee
codes, some codes that are "close" to being perfect are considered; see e.g. 
\cite{Bose}, where quasi-perfect Lee codes have been introduced. We show, by
means of our new approach, that these codes are a natural extension of
perfect Lee codes. So far quasi-perfect Lee codes have been found only for $%
n=2.$ Using our new approach we construct first quasi-perfect Lee codes for $%
n>2.$ On the other hand we prove that, for each $n\geq 3,$ there are at most
finitely many values of $e$ for which there exists a quasi-perfect $e$%
-error-correcting Lee code in $%
\mathbb{Z}
^{n}.$\medskip

\section{Terminology and Basic Concepts}

\noindent Throughout the paper we will use $%
\mathbb{Z}
^{n}$ both for the $n$-fold Cartesian product of the set $%
\mathbb{Z}
$ of integers and for the abelian(component-wise) additive group on $%
\mathbb{Z}
^{n}$. It will always be clear from the context which of the two we have in
mind.\ Because of the coding theory background the elements of $%
\mathbb{Z}
^{n}$ will be called words. The Lee distance (=the Manhattan distance) $\rho
_{L}(v,w)$ of two words $v=(v_{1},v_{2},...,v_{n}),$ $w=(w_{1},...,w_{n})$
in $%
\mathbb{Z}
^{n}$ is given by $\rho _{L}(v,w)=$ $\sum\limits_{i=1}^{n}\left\vert
v_{i}-w_{i}\right\vert $. By $S_{n,r}$ we denote the Lee sphere of radius $r$
in $%
\mathbb{Z}
^{n}$ centered at the origin $O;$ that is, $S_{n,r}=\{w;\rho _{L}(O,w)\leq
r\}$. The Lee sphere of radius $r$ in $%
\mathbb{R}
^{n},$ denoted $L_{n,r},$ is the union of unit cubes centered at words in $%
S_{n,e}.$ Further, $e_{i}$ will stand for the word $(0,...,0,1,0,...,0)$
with $i$-th coordinate equal to $1,$ and we will use $[a,b]$ as a shorthand
for all integers $k,a\leq k\leq b;$ $\left[ a,b\right] $ will also be called
a segment or an interval in $%
\mathbb{Z}
.$\medskip\ 

\noindent A code $\mathcal{L}$ in $%
\mathbb{Z}
^{n}$ is a subset of $%
\mathbb{Z}
^{n}.$ If a code $\mathcal{L}$ is at the same time a lattice then $\mathcal{L%
}$ is called a linear code. Linear codes play a special role as in this case
there is a better chance for the existence of an efficient decoding
algorithm. A code $\mathcal{L}$ is called a perfect $e$-error-correcting Lee
code in $%
\mathbb{Z}
^{n},$ denoted $PL(n,e),$ if (i) $\rho _{L}(v,w)\geq 2e+1$ for every $v,w\in 
\mathcal{L}{\normalsize ;}$ and (ii) every word $v\in 
\mathbb{Z}
^{n}$ is at Lee distance at most $e$ from a unique codeword in $\mathcal{L}%
{\normalsize .}$ Another way how to introduce a $PL(n,e)$ code is by means
of a tiling. Let $V$ be a subset of $%
\mathbb{Z}
^{n}$. \ By a copy of $V$ we mean a translation $V+x=\{v+x,v\in V\}$ of $V,$
where $x\in 
\mathbb{Z}
^{n}.$ A collection $\mathcal{T=\{}V+l;l\in \mathcal{L}\},$ $\mathcal{%
L\subset }%
\mathbb{Z}
^{n},$ of copies of $V$ constitutes a tiling of $%
\mathbb{Z}
^{n}$ by $V$ if $\mathcal{T}\,$forms a partition of $%
\mathbb{Z}
^{n}.$ $\mathcal{T}$ \ is called periodic (lattice) tiling if $\mathcal{L}$
is periodic (forms a lattice). Clearly, a set $\mathcal{L}$ is a $PL(n,e)$
code if and only if $\mathcal{\{}S_{n,e}+l;l\in \mathcal{L}\}$ constitutes a
tiling of $%
\mathbb{Z}
^{n}$ by Lee spheres $S_{n,e}$.\medskip\ 

\noindent If the condition (ii) in the definition of the $PL(n,e)$ code $%
\mathcal{L}$ is relaxed to: (iia) every word $v\in 
\mathbb{Z}
^{n}$ is at Lee distance at most $e+1$ from at least one codeword in $%
\mathcal{L}$, then $\mathcal{L}$ is called a quasi-perfect $e$%
-error-correcting Lee code, a $QPL(n,e)$ code. An efficient decoding
algorithm for quasi-perfect codes has been given in \cite{H}.\medskip

\section{Golomb - Welch Conjecture}

\noindent In this section we present a short account of the state of the art
in the Golomb-Welch conjecture, and describe various approaches how the
conjecture has been tackled so far.\medskip

\noindent For $n=2$ and all $e\geq 1,$ and for $e=1$ and all $n\geq 1,$ $%
PL(n,e)$ codes have been constructed by several authors, see \cite{Stein}.
Golomb and Welch \cite{GW} conjectured that:

\begin{conjecture}
There is no $PL(n,e)$ code for $n\geq 3$ and $e>1.$
\end{conjecture}

\noindent It is shown in \cite{GW} that for each $n$ there exists $e_{n},$
not specified in \cite{GW}, so that for all $e>e_{n}$ there is no $PL(n,e)$
code. To prove this statement the authors use a clever geometric argument;
we will describe it in detail in the proof of Theorem \ref{9}.
Unfortunately, the given type of argument cannot be used in the case when $e$
is relatively small to $n$.\medskip\ 

\noindent Another type of a geometric argument, also making use of tiling by
Lee spheres $L_{n,e},$ has been used in \cite{Gra} to settle the G-W
conjecture for $n=3$ and all $e>1.$ It is an elegant "picture says it all"
proof that, unfortunately, cannot be extended to a higher dimension. Later 
\v{S}pacapan \cite{Sca}, whose proof is computer aided, showed the
non-existence of a $PL(n,e)$ code for $n=4,$ and all $e>1$. His method
cannot be extended even to $n=5$ as the number of cases needed to be checked
in his approach grows too rapidly. It is proved in \cite{Ho} that there is
no $PL(n,e)$ code for $3\leq n\leq 5,$ and all $e>1.$ The only other value
of parameters for which the Golomb-Welch conjecture is known to be true is $%
n=$ $6$ and $e=2,$ see \cite{Hor}. The last two results are proved using an
algebraic/counting argument showing that a $PL(n,e)$ code does not exist
even in a "local sense". Unfortunately, this method is not suitable for
bigger values of $n$.\medskip

\noindent In some papers the non-existence of special types of $PL(n,e)$
codes is proved. We mention here only two of them. Post \cite{P} showed that
there is \ no periodic $PL(n,e)$ code for $3\leq n\leq 5,e\geq n-2,$ and for 
$n\geq 6,$ and $e\geq \frac{\sqrt{2}}{2}n-\frac{1}{4}(3\sqrt{2}-2).$ To
prove it Post used generating functions. Also this method is unsuitable for
small values of $e$. For $e\geq n\geq 3,$ Post's result has been improved in 
\cite{Sca1}, where it is shown that there is no so called optimal Lee-type
local structure for given parameters.\medskip

\noindent There are several reformulations of the Golomb-Welch conjecture.
One is in terms of the perfect domination set in a graph isomorphic to
Cartesian product of cycles, while a reformulation of the conjecture in
terms of circulant graphs appears in \ \cite{C}. So far these two
reformulations have not been helpful in progressing with the Golomb-Welch
conjecture.\medskip

\noindent At the end of this section we briefly describe two extensions of
the Golomb-Welch conjecture. For a detailed account we refer the reader to 
\cite{HB1}. Diameter-$d$ perfect codes have been introduced in the Hamming
scheme by Ahlswede et al. in \cite{A1}, while Etzion \cite{E} extended the
notion to Lee metric. Since, for $d$ odd, a diameter-$d$ perfect Lee code is
a $PL(n,\frac{d-1}{2})$ code as well, these codes constitute a
generalization of perfect error-correcting Lee codes. Therefore the
conjecture stated by Etzion in \cite{E} is an extension of the G-W
conjecture. A further extension of Etzion's conjecture to the \textit{%
perfect distance-dominating set} in a graph $G$ has been stated in \cite{ADH}%
. Unfortunately, as with the mentioned reformulations of the G-W conjecture,
the two extensions have not contributed yet to the solution of the
conjecture. A tiling constructed by Minkowski \cite{M1} provides an
exception to both extensions of the G-W conjecture. However, we do not
believe that this indicates that there might be an exception to the G-W
conjecture as well.

\section{Embedding Abelian Groups}

\noindent In this section we initiate the study of a new invariant of
abelian groups. We show how this invariant is related to the G-W
conjecture.\medskip

\noindent Let $G$ be a finite abelian group and $\phi :%
\mathbb{Z}
^{n}\rightarrow G$ be a homomorphism. For $g\in \phi (%
\mathbb{Z}
^{n})$ we set $\pi (n,G,\phi ,g)=\min \{\rho _{L}(x,O);$ where $\phi
(x)=g\}, $ and say that $g$ is embedded at the minimum distance $\pi
(n,G,\phi ,g)$. If $\phi $ is surjective, the embedding number of $G$ into $%
\mathbb{Z}
^{n\text{ }}$with respect to $\phi $ is defined to be the number $\pi
(n,G,\phi )=\dsum\limits_{g\in G}\pi (n,G,\phi ,g),$ otherwise we put $\pi
(n,G,\phi )=\infty $. The embedding number $\pi (n,G)$ of $G$ in $%
\mathbb{Z}
^{n}$ is set to be $\min_{\phi }\pi (n,G,\phi )$ where the minimum is taken
over all homomorphisms $\phi :%
\mathbb{Z}
^{n}\rightarrow G.$ Finally, for each $k>0,$ we set $\pi (n,k)=\min_{G}\pi
(n,G),$ where the minimum runs over all abelian groups of order $k.$ We note
that the value of $\pi (n,G)$ is invariant under the group isomorphism;
i.e., if $G\simeq H$ then $\pi (n,G)=\pi (n,H)$.\medskip

\noindent In the following example we illustrate the definition of the
embedding number by means of the cyclic group $Z_{16}.$\medskip

\noindent \textit{Example}. Consider a homomorphism $\phi :Z^{2}\rightarrow
Z_{16}$ given by $\phi (e_{1})=1,$ and $\phi (e_{2})=5.$ Then $\pi
(2,Z_{16},\phi ,g)=0$ for $g=0,$ $\pi (2,Z_{16},\phi ,g)=1$ for $%
g=1,5,11,15; $ $\pi (2,Z_{16},\phi ,g)=2$ for $g=2,4,6,10,12,14;$ $\pi
(2,Z_{16},\phi ,g)=3$ for $g=3,7,9,13,$ and $\pi (2,Z_{16},\phi ,g)=4$ for $%
g=8.$ Therefore, $\pi (2,Z_{16},\phi )=0\cdot 1+1\cdot 4+2\cdot 6+3\cdot
4+4\cdot 1=32.$ The homomorphism $\phi $ is illustrated in Fig.1. The Lee
sphere $S_{2,2}$ is bounded by a thick line, while $S_{2,3}$ is bounded by a
double line. The numbers given there are values of $\phi ((x,y)\dot{)}\in
Z_{16}$ at the given point of $Z^{2}.$ The elements in bold font and
underlined are embeddings of elements of $Z_{16}\,\ $\ at the minimum
distance from the origin. If there were more embeddings of an element at the
minimum distance we have picked one of them at random; e.g., there are two
embeddings of $10\in Z_{16}$ at the minimum Lee distance $2,$ and of $13\in
Z_{16}$ at the minimum Lee distance $3.$\medskip\ 

\FRAME{fhFU}{2.0358in}{2.4543in}{0pt}{\Qcb{Homomorphism $\protect\phi %
:Z^{2}\rightarrow Z_{16}$}}{}{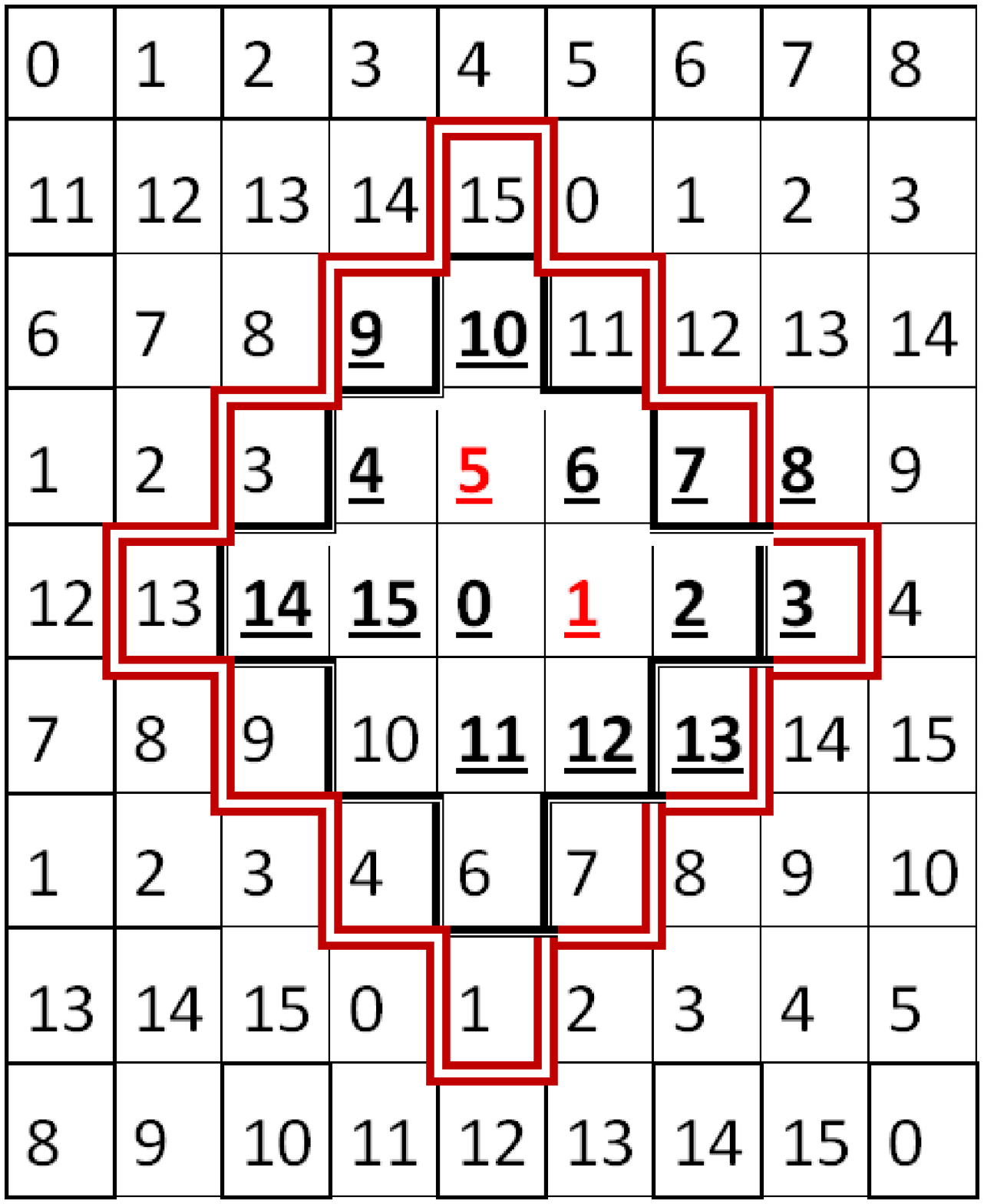}{\special{language "Scientific
Word";type "GRAPHIC";maintain-aspect-ratio TRUE;display "USEDEF";valid_file
"F";width 2.0358in;height 2.4543in;depth 0pt;original-width
7.0249in;original-height 8.489in;cropleft "0";croptop "1";cropright
"0.9996";cropbottom "0";filename 'obr3.eps';file-properties "XNPEU";}}

\noindent The (hypothetically) best value of $\pi (2,Z_{16})$ would be
attained by a homomorphism, if any, $\phi :Z^{2}\rightarrow Z_{16},$ with
the property that there are $\left\vert S_{2,1}\right\vert -\left\vert
S_{2,0}\right\vert =5-1=4$ elements $g$ of $Z_{16}$ with $\pi (2,Z_{16},\phi
)=1;$ $\left\vert S_{2,2}\right\vert -\left\vert S_{2,1}\right\vert =13-5=8$
elements $g$ of $Z_{16}$ with $\pi (2,Z_{16},\phi )=2;$ and finally $%
\left\vert Z_{16}\right\vert -\left\vert S_{2,3}\right\vert =16-13=3$
elements of $g$ with $\pi (2,Z_{16},\phi ,g)=3.$ Then, in total, $\pi
(2,Z_{16},\phi )=0\cdot 1+1\cdot 4+2\cdot 8+3\cdot 3=29.$ It will be shown
in the next theorem that such an embedding for $Z_{16}$ is attained by the
homomorphism $\phi $ given by $\phi (e_{1})=2$ and $\phi (e_{2})=3.$ So, $%
\pi (2,Z_{16})=29.$ As the lower bound applies to any abelian group of order 
$16,$ we also have $\pi (2,16)=29.$\medskip

\noindent To be able to show how the above introduced notion of group
embeddings relates to the G-W conjecture we first present a lower bound on $%
\pi (n,k)$ and then state a theorem proved in \cite{HB1}.\medskip\ 

\noindent Let $n,k\geq 1.$ Then there is a uniquely determined number $r$ so
that $\left\vert S_{n,r}\right\vert \leq k<\left\vert S_{n,r+1}\right\vert .$
To facilitate our discussion we set 
\[
f(n,k)=[\dsum\limits_{1\leq i\leq r}i(\left\vert S_{n,i}\right\vert
-\left\vert S_{n,i-1}\right\vert )]+(r+1)(k-|S_{n,r}|). 
\]

\begin{theorem}
\label{7}Let $k,n\geq 1.$ Then $\pi (n,k)\geq f(n,k).$ Moreover, for $%
\left\vert S_{n,r}\right\vert <k~$(for $\left\vert S_{n,r}\right\vert =k),$ $%
\pi (n,k)=f(n,k)$ if and only if there is an abelian group $G$ of order $k$
and a homomorphism $\phi :%
\mathbb{Z}
^{n}\rightarrow G$ such that the restriction of $\phi $ to $S_{n,r}$ is
injective and the restriction of $\phi $ to $S_{n,r+1}$ is surjective (the
restriction of $\phi $ to $S_{n,r}$ is a bijection).
\end{theorem}

\noindent We will say that a number $k>0$ (an abelian group $G$ \ of order $%
k)$ has an optimal embedding in $%
\mathbb{Z}
^{n}$ if $\pi (n,k)=f(n,k)$ (if $\pi (n,G)=f(n,k)$)$.$\medskip

\begin{proof}
Denote by $G_{d}$ the set $\{g;g\in G,$ such that $\pi (n,G,\phi ,g)=d\}$,
and, for $d\leq r,$ $\varepsilon _{d}=(\left\vert S_{n,d}\right\vert
-\left\vert S_{n,d-1}\right\vert )-\left\vert G_{d}\right\vert .$ Since
there are in $%
\mathbb{Z}
^{n}$ exactly $\left\vert S_{n,d}\right\vert -\left\vert
S_{n,d-1}\right\vert $ words at distance $d$ from the origin, we have $%
\left\vert G_{d}\right\vert \leq \left\vert S_{n,d}\right\vert -\left\vert
S_{n,d-1}\right\vert ,$ and thus $\varepsilon _{d}\geq 0.$ We get $\pi
(n,G,\phi )=\dsum\limits_{g\in G}\pi (n,G,\phi ,g)=\dsum\limits_{d\geq
0}d\left\vert G_{d}\right\vert =\dsum\limits_{0\leq d\leq r}d\left\vert
G_{d}\right\vert +\dsum\limits_{d>r}d\left\vert G_{d}\right\vert =$ $%
(\dsum\limits_{d\leq r}d(\left\vert S_{n,d}\right\vert -\left\vert
S_{n,d-1}\right\vert -\varepsilon _{d}))+(r+1)(\left\vert G\right\vert
-\left\vert S_{n,r}\right\vert +\dsum\limits_{d\leq r}\varepsilon
_{d})+\dsum\limits_{d\geq r+2}(d-r-1)\left\vert G_{d}\right\vert $ since $%
\dsum\limits_{d>r}\left\vert G_{d}\right\vert =\left\vert G\right\vert
-\left\vert S_{n,r}\right\vert +\dsum\limits_{d\leq r}\varepsilon _{d}.$
Therefore, 
\begin{equation}
\pi (n,G,\phi )=f(n,r)+\dsum\limits_{d\leq r}(r+1-d)\varepsilon
_{d}+\dsum\limits_{d\geq r+2}(d-r-1)\left\vert G_{d}\right\vert .  \label{B}
\end{equation}

\noindent By (\ref{B}), $\pi (n,G,\phi )\geq f(n,r)$ for all homomorphisms $%
\phi .$ There is an equality in (\ref{B}) iff $\varepsilon _{d}=0$ for all $%
d\leq r,$ and $\left\vert G_{d}\right\vert =0$ for all $d>r+1;$ i.e., iff
the restriction of $\phi $ to $S_{n,r}$ is injective, and the restriction of 
$\phi $ to $S_{n,r+1}$ is surjective. For $k=S_{n,r}$ this necessary and
sufficient condition translates to $\phi $ is a bijection on $%
S_{n,r.}\medskip $
\end{proof}

\noindent The following theorem has been stated in \cite{HB1}

\begin{theorem}
\label{6} Let $V$ be a subset of $%
\mathbb{Z}
^{n}.$ Then there is a lattice tiling of $%
\mathbb{Z}
^{n}$ by $V$ if and only if there is an abelian group $G$ of order $%
\left\vert V\right\vert ,$ and a homomorphism $\phi :%
\mathbb{Z}
^{n}\rightarrow G$ so that the restriction of $\phi $ to $V$ is a bijection.
\end{theorem}

\noindent Combining the above two theorems for $V=S_{n.e}$ yields:

\begin{corollary}
\label{CC}There exists a linear $PL(n,e)$ code if and only if there is an
optimal embedding of the number $\left\vert S_{n,e}\right\vert $ in $%
\mathbb{Z}
^{n}.$
\end{corollary}

\noindent In turn, we get a reformulation of the G-W conjecture in the case
of linear codes:

\begin{conjecture}
\label{10} The number $\left\vert S_{n,e}\right\vert $ does not have an
optimal embedding in $%
\mathbb{Z}
^{n}$ for $n\geq 3$ and $e>1.$
\end{conjecture}

\noindent The following theorem constitutes another main results of the
paper.

\begin{theorem}
\label{11}Each $k\geq 1$ has an optimal embedding in $%
\mathbb{Z}
^{2}.$ In particular, for each $k\geq 1,$ the cyclic group $%
\mathbb{Z}
_{k}$ has an optimal embedding in $%
\mathbb{Z}
^{2}$.
\end{theorem}

\begin{proof}
The intersection $l_{r,m}$ of the sphere $S_{2,r}$ with the line $x+y=m$ is
non-empty if and only if $-r\leq m\leq r$, and $l_{r,m}$ comprises points $%
(x,y)$ in $%
\mathbb{Z}
^{2}$ with $\left\lceil \frac{m-r}{2}\right\rceil \leq x\leq \left\lfloor 
\frac{m+r}{2}\right\rfloor ,$ $\left\lceil \frac{m-r}{2}\right\rceil \leq
y\leq \left\lfloor \frac{m+r}{2}\right\rfloor ,$ and $x+y=m$. We split the
sphere $S_{2,r}$ into the upper part $\overline{S}_{2,r}=\{(x,y);$ $(x,y)\in
S_{2,r},$ and $x+y=m,$ where $0<m\leq r,$ or $m=0$ and $x<0\}$, the lower
part $\underline{\emph{S}}_{2,r}=S_{2,r}-\overline{S}_{2,r}-\{(0,0)\}$, and
the origin. As $(x,y)\in \overline{S}_{2,r}$ implies $(-x,-y)\in \underline{%
\emph{S}}_{2,r}$ we get $\phi (\underline{\emph{S}}_{2,r})=-\phi (\overline{S%
}_{2,r})$ for each homomorphism $\phi $ on $%
\mathbb{Z}
^{2}.$\medskip

\noindent Let $\Phi :%
\mathbb{Z}
^{2}\rightarrow 
\mathbb{Z}
$ be a homomorphism given by $\Phi (e_{1})=r,$ and $\Phi (e_{2})=r+1.$ Since 
$\Phi ((x,y))-\Phi ((x+1,y-1))=$ $rx+(r+1)y-r(x+1)-(r+1)(y-1)=1,$ we get
that $\Phi (l_{r,m})$ is a segment in $%
\mathbb{Z}
,$ and $\Phi $ is decreasing on $l_{r,m}$ in the $x$-coordinate. To prove
that $\Phi (l_{r,m})$ and $\Phi (l_{r,m+1})$ constitute two consecutive
segments it suffices to verify (we leave it to the reader) that for all $%
r,m, $ the value of $\Phi $ at the point $(x,y)$ in $l_{r,m}$ with the
smallest value of $x,$ is by one smaller than the value of $\Phi $ at the
point $(x,y) $ in $l_{r,m+1}$ with the largest value of $x,$ i.e., that $%
\Phi ((m-\left\lfloor \frac{m+r}{2}\right\rfloor ,\left\lfloor \frac{m+r}{2}%
\right\rfloor ))+1=$ $\Phi ((m+1-\left\lceil \frac{m+1-r}{2}\right\rceil
,\left\lceil \frac{m+1-r}{2}\right\rceil ))$. This in turn implies that $%
\Phi (S_{2,r})$ is the union of consecutive segments and the restriction of $%
\Phi $ to $S_{2,r}$ is an injection. Hence, $\left\vert \phi
(S_{2,r})\right\vert =\left\vert S_{2,r}\right\vert =2r^{2}+2r+1.$ In
addition, we have $\Phi ((0,0))=0,$ and therefore $\Phi (\overline{S}%
_{2,r})=[1,r(r+1)],$ and $\Phi (\underline{S}_{2,r})=[-r(r+1),-1]$.\medskip

\noindent Now we are ready to prove the statement of the theorem. First, let

(i) $k=\left\vert S_{2,r}\right\vert =2r^{2}+2r+1$ for some $r.$ In this
case the statement that $k$ has an optimal embedding in $%
\mathbb{Z}
^{2}$ is equivalent to the statement that there is a tiling of $%
\mathbb{Z}
^{2}$ by Lee spheres $L_{2,r};$ this has been shown by several authors, see
e.g. \cite{GW}. It is not difficult to see that such a tiling is unique, up
to a symmetry. Thus, in fact we show that the unique tiling is a lattice
one, and moreover, the group associated with the lattice is the cyclic
group. Consider the homomorphism $\phi :%
\mathbb{Z}
^{2}\rightarrow 
\mathbb{Z}
_{k}$ given by $\phi (e_{1})=r,\phi (e_{2})=r+1$. Clearly, $\phi (u)=\Phi
(u)(\func{mod}k)$ for each $u\in 
\mathbb{Z}
^{2},$ and thus $\phi ((x,y))=\Phi ((x,y))$ for $(x,y)\in \overline{S}%
_{2},_{r}\,$ as $\Phi (\overline{S}_{2,r})=[1,r(r+1)]\,$\ and $r(r+1)<k.$ In
addition, $\phi $ is injective on $\overline{S}_{2},_{r}$ because $\Phi $
is. As $\phi $ is a homomorphism, $\phi $ is injective on $\underline{S}%
_{2,r}\,\ $\ and $\phi (\underline{S}_{2,r})=[-r(r+1),-1](\func{mod}k)=$ $%
[-r(r+1)+k,-1+k]=$ $[r(r+1)+1,2r^{2}+2r].$ In aggregate, $\phi
(S_{2,r})=\phi (\overline{S}_{2,r})\cup \phi (\underline{S}_{2,r})\cup \phi
((0,0))=[0,2r^{2}+2r]=[0,k-1].$ Thus, the restriction of $\phi $ to $S_{2,r}$
is both an injection and a surjection. Therefore, by Theorem \ref{7}, $k$
has an optimal embedding in $%
\mathbb{Z}
^{2}.$\medskip

(ii) Let $r$ be so that $\left\vert S_{2,r}\right\vert <k<\left\vert
S_{2,r+1}\right\vert ;$ that is, $2r^{2}+2r+1<k<2r^{2}+6r+5.$ We split the
proof into two cases.\medskip

(iia) $\left\vert S_{2,r}\right\vert =2r^{2}+2r+1<k\leq 2r^{2}+4r.$

\noindent Let $\phi :%
\mathbb{Z}
^{2}\rightarrow 
\mathbb{Z}
_{k}$ be the same homomorphism as above.\ Then $\phi (u)=\Phi (u)(\func{mod}%
k),$ and $\phi (\overline{S}_{2,r})=[1,r(r+1)],$ $\phi $ is injective on $%
\overline{S}_{2},_{r}$, and $\phi (\underline{\emph{S}}_{2,r})=-\phi (%
\overline{S}_{2,r})=[-r(r+1),-1](\func{mod}k)=[-r(r+1)+k,-1+k]$. As $%
k>2r^{2}+2r+1,$ we have $-r(r+1)+k>r(r+1),$ and this implies $\phi (%
\underline{\emph{S}}_{2,r})\cap \phi (\overline{S}_{2,r})=\emptyset .$ Thus $%
\phi $ is an injection on $S_{2,r}$. To finish the proof we need to show
that the restriction of $\phi $ to $S_{2,r+1}$ is a surjection. $\phi
(X)=\Phi (X)(\func{mod}k)$ yields $\phi (\overline{S}_{2,r+1})\supset \Phi (%
\overline{S}_{2,r})\cup \Phi (l_{r+1,r+1})=[1,r(r+1)]\cup \lbrack \Phi
((r+1,0)),\Phi ((0,r+1))]=$

$[1,r(r+1)]\cup \lbrack r(r+1),(r+1)(r+1)]=[1,(r+1)^{2}].$ Further, $\phi (%
\underline{\emph{S}}_{2,r})=-\phi (\overline{S}_{2,r})$ implies $\phi (%
\underline{\emph{S}}_{2,r})\supset \lbrack -(r+1)^{2},-1](\func{mod}k),$
that is, $\phi (\underline{\emph{S}}_{2,r})\supset \lbrack
-(r+1)^{2}+k,-1+k].$ However, it is $(r+1)^{2}>k-(r+1)^{2}$ because in this
case $k\leq 2r^{2}+4r.$ Therefore $\phi (S_{2,r+1})=[0,k-1],$ i.e., $\phi $
is a surjection on $S_{2,r+1}.$\medskip

(iib) $2r^{2}+4r+1\leq k<2r^{2}+6r+5.$

\noindent Consider a homomorphism $\Phi ^{\prime }:%
\mathbb{Z}
^{2}$\noindent $\rightarrow 
\mathbb{Z}
$ given by $\Phi ^{\prime }(e_{1})=r+1,\Phi ^{\prime }(e_{2})=r+2.$
Translating the results obtained above for the homomorphism $\Phi $ $\ $into
the language of $\Phi ^{\prime }$ we get that the restriction of $\Phi
^{\prime }$ to $S_{2,r+1}$ is injective, and $\Phi ^{\prime }(\overline{S}%
_{2,r+1})=[1,(r+1)(r+2)],$ $\Phi ^{\prime }(\underline{S}%
_{2,r+1})=[-(r+1)(r+2),-1].$ Further, we have $\max \Phi ^{\prime }(%
\overline{S}_{2,r})=\Phi ^{\prime }((0,r))=r(r+2)$.\medskip

\noindent Let $\phi :%
\mathbb{Z}
^{2}\rightarrow 
\mathbb{Z}
_{k}$ be a homomorphism given by $\phi (e_{1})=r+1,\phi (e_{2})=r+2$. Hence, 
$\phi (u)=\Phi ^{\prime }(u)(\func{mod}k),$ which in turn implies $\phi (%
\overline{S}_{2,r+1})=\Phi (\overline{S}_{2,r+1})=[1,(r+1)(r+2)]\supset
\lbrack 1,\frac{k-1}{2}].$ In addition, $\phi (\underline{S}_{2,r+1})=-\phi (%
\overline{S}_{2,r+1})\supset \lbrack -\frac{k-1}{2},-1](\func{mod}k)=[-\frac{%
k-1}{2}+k,-1+k]=[\frac{k+1}{2},k-1].$ In aggregate, $\phi (S_{2,r+1})=\phi (%
\overline{S}_{2,r+1})\cup \phi (\underline{S}_{2,r+1})\cup \phi
((0,0))\supset \lbrack 1,\frac{k-1}{2}]\cup \lbrack \frac{k+1}{2},k-1]\cup
\{0\}=[0,k-1].$ Thus, the restriction of $\phi $ to $S_{2,r+1}$ is
surjective. Now we prove that the restriction of $\phi $ to $S_{2,r}$ is
injective. We recall that the restriction of $\Phi ^{\prime }$ to $S_{2,r}$
is injective and $\max \Phi ^{\prime }(\overline{S}_{2,r})=\Phi ^{\prime
}(0,r)=r(r+2)$. Thus $\Phi ^{\prime }(\overline{S}_{2,r})\subset \lbrack
1,r(r+2)].$ Therefore $\phi $ is injective on $\overline{S}_{2,r}$, $\phi (%
\overline{S}_{2,r})\subset \lbrack 1,\frac{k-1}{2}].$ Further, $\phi $ is
injective on $\underline{S}_{2,r},$ and $\phi (\underline{S}_{2,r})=-\phi (%
\overline{S}_{2,r})\subset \lbrack -\frac{k-1}{2},-1](\func{mod}k)=$ $[\frac{%
k-1}{2}+k,-1+k]=[\frac{k+1}{2},k-1].$ Hence, $\phi (\overline{S}_{2,r})\cap
\phi (\underline{S}_{2,r})$ is empty.\medskip
\end{proof}

\noindent Now we prove that also in this case the results for $n\geq 3$
dramatically differ from those for $n=2.$

\begin{theorem}
\label{9} For each $n\geq 3,$ there is a $k_{n}$ so that no $k\geq k_{n}$
has an optimal embedding in $%
\mathbb{Z}
^{n}.$
\end{theorem}

\begin{proof}
This proof uses ideas developed in \cite{GW}. Let $P_{n,r}$ be the smallest
convex polytope containing the $2n$ points $\pm (r+\frac{1}{2}%
)e_{i},i=1,...,n.$ In other words, $P_{n,r}$ is the smallest convex polytope
containing $2n$ center points of ($n-1$)-dimensional extremal hyperfaces of
the Lee sphere $L_{n,r}$. Thus, $P_{2,r}$ is a square while $P_{3,r}$ is a
regular octahedron. For the $n$-dimensional hypervolume of the regular
polytope $P_{n,r}$ we have: $V(P_{n,r})=\frac{(2r+1)^{n}}{n!}.$ If there was
a tiling of $%
\mathbb{R}
^{n}$ by Lee spheres $L_{n,r}$ then this tiling would induce a packing of $%
\mathbb{R}
^{n}$ by regular polytopes $P_{n,r}$.\medskip

\noindent Assume now that an integer $k$ has an optimal embedding in $%
\mathbb{Z}
^{n}.$ Then, by Theorem \ref{7}, there is a homomorphism $\phi :%
\mathbb{Z}
^{n}\rightarrow G,$ an abelian group of order $k,$ so that the restriction
of $\phi $ to $S_{n,r}$ is injective and the restriction of $\phi $ to $%
S_{n,r+1}$ is surjective. Therefore we are able to choose a set $\ K$ of $k$
points in $%
\mathbb{Z}
^{n}$ so that $S_{n,r}\subseteq K\subset S_{n,r+1}$ and the restriction of $%
\phi $ on $K$ is a bijection. By Theorem \ref{6} $\phi $ induces a lattice
tiling of $%
\mathbb{Z}
^{n}$ by copies of $K.$ This in turn implies that there is a lattice tiling $%
\mathcal{T}$ of $%
\mathbb{R}
^{n}$ by the tile $T_{K}$ comprising unit cubes centered at points in $K.$
Since $S_{n,r}\subseteq K,$ the Lee sphere $L_{n,r}$ is a subset of $T_{K}.$
Therefore, tiling $\mathcal{T}$ induces a packing of $%
\mathbb{R}
^{n}$ by Lee spheres $L_{n,r},$ which in turn induces a packing of $%
\mathbb{R}
^{n}$ by polytopes $P_{n,r}.$\medskip

\noindent It is known, see \cite{Cox}, that regular polytopes $P_{n,k}$ do
not tile $%
\mathbb{R}
^{n}.$ Further, cf. \cite{GW}, if a polytope does not tile \ $%
\mathbb{R}
^{n},$ then the packing efficiency $\alpha $ of $%
\mathbb{R}
^{n}$ by this polytope is strictly less than $1$. As the packing of $%
\mathbb{R}
^{n}$ by copies of $P_{n,r}$ has been induced by tiling $\mathcal{T}$ of $%
\mathbb{R}
^{n}$ by the cluster of unit cubes $T_{K}$, $\frac{V(P_{n,r})}{V(T_{K})}\leq
\alpha .$ Therefore, there is no tiling $\mathcal{T}{\normalsize \ }$of $%
\mathbb{R}
^{n}$ by a tile $T_{k}$ for 
\begin{equation}
\frac{V(P_{n,r})}{V(T_{K})}>\alpha .  \label{13}
\end{equation}%
However, for the volume $V(T_{k})$ we have $V(L_{n,r})\leq
V(T_{k})<V(L_{n,r+1}).$ Therefore, for $n$ fixed$,$

\[
\lim_{r\rightarrow \infty }\frac{V(P_{n,r})}{V(L_{n,r})}=\lim_{r\rightarrow
\infty }\frac{V(P_{n,r})}{V(T_{K})}=\lim_{r\rightarrow \infty }\frac{%
V(P_{n,r})}{V(L_{n,r+1})}=1 
\]

\noindent as the coefficient at the leading term of both polynomials $%
V(L_{n,r})$ and $V(L_{n,r+1})$ is $\frac{2^{n}}{n!}r^{n}.$ Here we recall,
see e.g. \cite{GW}, that the volume of $L_{n,r}$ equals $\dsum\limits_{k\geq
0}2^{k}\binom{n}{k}\binom{r}{k}$. Thus, there is a $k_{n}$ so that, for all $%
k>k_{n},$ we have $\frac{V(P_{n,r})}{V(T_{K})}>\alpha .$ Hence, by (\ref{13}%
), for $k>k_{n},$ there is no tiling of $%
\mathbb{R}
^{n}$ by $T_{k},$ that is, there is no tiling of $%
\mathbb{Z}
^{n}$ by the set $K$ in this case. This in turn implies that $k=\left\vert
K\right\vert $ does not have an optimal embedding in $%
\mathbb{Z}
^{n}.$
\end{proof}

\noindent At the end of this section we prove the non-existence of a linear $%
PL(n,e)$ code for some new values of parameters.

\begin{theorem}
There is no linear $PL(n,2)$ code for $7\leq n\leq 12.$
\end{theorem}

\begin{proof}
By Corollary \ref{CC} it suffices to prove that, for $n=7,...,??,$ the
number $k_{n}=\left\vert S_{n,2}\right\vert =2n^{2}+2n+1$ does not have an
optimal embedding in $%
\mathbb{Z}
^{n}.$ For $\ n=7,...,12,$ $k_{n}$ is a square free number. Thus each
abelian group of the order $k_{n}$ is isomorphic to the cyclic group $%
\mathbb{Z}
_{k_{n}}$. We need to show that there is no homomorphism $\phi :%
\mathbb{Z}
^{n}\rightarrow 
\mathbb{Z}
_{k_{n}}$ such that the restriction of $\phi $ to $S_{n,2}$ is bijective.
Each homomorphism $\phi :%
\mathbb{Z}
^{n}\rightarrow 
\mathbb{Z}
_{k_{n}}$ is determined by the values of $\phi (e_{i}),i=1,...,n,$ and $%
\left\vert \phi (S_{n,2})\right\vert =\left\vert \{\pm \phi (e_{i}),\pm \phi
(e_{i})\pm \phi (e_{j});1\leq i\leq j\leq n\}\right\vert <\left\vert
S_{n,2}\right\vert $ if $\phi $ is not a bijection on $S_{n,2}.$ Hence, it
is sufficient to show that for each $n$-tuple $(g_{1},...,g_{n})$ of
elements in $%
\mathbb{Z}
_{k_{n}}$ 
\begin{equation}
\left\vert \{\pm g_{i},\pm g_{i}\pm g_{j};1\leq i\leq j\leq n\}\right\vert
<\left\vert S_{n,2}\right\vert .  \label{N}
\end{equation}%
\noindent This can be proved by a \textit{brute force} computer test, where
all $(k_{n})^{n}\approx (2n^{2})^{n}$ $n$-tuples of elements in $%
\mathbb{Z}
_{k_{n}}$ are shown to satisfy (\ref{N}). In what follows we exhibit a way
how to substantially reduce the computational complexity of the test.\medskip

\noindent Assume that there exists a homomorphism $\phi :%
\mathbb{Z}
^{n}\rightarrow 
\mathbb{Z}
_{k_{n}}$ such that the restriction of $\phi $ to $S_{n,2}$ is a bijection.
Then there would have to be such a homomorphism $\phi ^{\prime }$ with $\phi
^{\prime }(e_{j})\leq \frac{k_{n}}{2}$ (if $\phi (e_{i})>\frac{k_{n}}{2}$ we
set $\phi ^{\prime }(e_{i})=-\phi (e_{i})$ ) for all $1\leq i\leq n,$ and
also $\phi ^{\prime }(i)<\phi ^{\prime }(j)$ for all $1\leq i<j\leq n.$
Therefore, we need to show that (\ref{N}) is satisfied by any of $\binom{%
k_{n}/2}{n}$ $n$-tuples $(g_{1},...,g_{n}),$ where $1\leq g(i)\leq \frac{%
k_{n}}{2},$ and $g(i)<g(j)$ for all $1\leq i<j\leq n$. Using the Stirling
formula we have $\binom{k_{n}/2}{n}\approx \frac{(en)^{n}}{\sqrt{2\pi n}}.$
To reduce the computational complexity even further we used a backtracking
algorithm that enables to check (\ref{N}) only for a portion of $\binom{%
k_{n}/2}{n}$ of $n$-tuples. The algorithm is based on the following two
simple observations: First, let $(g_{1},...,g_{n})$ be a $n$-tuple such that
for some $m<n$ 
\begin{equation}
\left\vert \{\pm g_{i},\pm g_{i}\pm g_{j};1\leq i\leq j\leq m\}\right\vert
<\left\vert S_{m,2}\right\vert ,  \label{NN}
\end{equation}%
then $(g_{1},...,g_{n})$ satisfies (\ref{N}) as well. Second, it suffices to
choose $g_{m}>g_{m-1}$ and $g_{m}\in T_{m}=[1,...,\frac{k_{n}-1}{2}]-\{\pm
g_{i},\pm g_{i}\pm g_{j};1\leq i\leq j\leq m-1\},1\leq m\leq n,$ as
otherwise (\ref{NN}) is trivially satisfied. The algorithm comprises $n$
nested cycles. In the $m$-th cycle we choose the $m$-th element of the tuple 
$(g_{1},...,g_{n}).$ It suffices to choose $g_{m}>g_{m-1}$ and $g_{m}\in
T_{m},$ see above. We test $(g_{1},...,g_{m})$ for (\ref{NN}). If (\ref{NN})
is satisfied, we replace $g_{m}$ by the next element from $T_{m},$ if none
is available, we backtrack to the previous cycle that chooses $g_{m-1}.$
Otherwise, if there is an equality in (\ref{NN}), in the next nested cycle
we choose $g_{m+1}>g_{m}$ from the set $T_{m+1}$.

\noindent Thus, in aggregate, the test (\ref{NN}) is performed at most 
\[
\frac{1}{n!}\dprod\limits_{m=1}^{n}\left\vert T_{m}\right\vert =\frac{1}{n!}%
\dprod\limits_{m=1}^{n}(\frac{\left\vert S_{n,2}\right\vert -1}{2}-\frac{%
\left\vert S_{m-1,2}\right\vert -1}{2})=\frac{1}{n!}\dprod%
\limits_{m=0}^{n-1}[(n^{2}+n)-(m^{2}+m)]= 
\]%
\[
\frac{1}{n!}\dprod\limits_{m=0}^{n-1}(n+m+1)(n-m)=\frac{(2n)!}{n!}\approx 
\sqrt{2}(\frac{4n}{e})^{n} 
\]
times. We note that the factor $\frac{1}{n!}$ stands for the fact that from
all $n$-tuples $(g_{1},...,g_{n})$ that differ only by the order of elements
we test only the one with $g(i)<g(j)$ for all $1\leq i<j\leq n$. Clearly,
this is a crude upper bound because if $(g_{1},...,g_{m})$ does not satisfy (%
\ref{NN}), then we skip testing all $t$-tuples $(g_{1},...,g_{m},...,g_{t})$
for each $t>m.$ The backtracking algorithm described above was \ used to
prove the statement.$\medskip $
\end{proof}

\begin{remark}
We point out that the backtracking algorithm described above can be used any
time when $k_{n}=\left\vert S_{n,2}\right\vert =2n^{2}+2n+1$ is a square
free number, and sufficient computing power is available. E.g., for all $%
n\leq 19,$ the number $k_{n}$ is square free. However, we have verified the
statement of the theorem only for $n\leq 12,$ as for the bigger values of $n$
the computation has not been feasible for our computer lab. We note that the
computation can easily be distributed over several machines as verifying (%
\ref{NN}) for distinct $n$-tuples is independent on each other. In fact we
used this distributed approach for all $n\geq 9$
\end{remark}

\noindent We end this section by a conjecture related to group embeddings.
If true, it would be much simpler to determine the value of $\phi (n,k).$

\begin{conjecture}
For each $n\geq 2,$ and $k>0,$ the value $\pi (n,k)$ is attained by the
cyclic group $%
\mathbb{Z}
_{k}.$
\end{conjecture}

\section{Quasi-Perfect Lee codes}

\noindent We start with a theorem that shows how the quasi-perfect codes
relate to optimal embeddings introduced in the previous section. In fact it
turns out that quasi-perfect Lee codes are a natural extension of perfect
Lee codes.

\begin{theorem}
\label{15} A linear $QPL(n,e)$ code exists if and only if there is a number $%
k,$ $\left\vert S_{n,e}\right\vert \leq k<\left\vert S_{n,e+1}\right\vert ,$
having an optimal embedding in $%
\mathbb{Z}
^{n}.$
\end{theorem}

\begin{proof}
Let $\phi :%
\mathbb{Z}
^{n}\rightarrow G,$ an abelian group $G$ of order $k,$ be a homomorphism so
that the restriction of $\phi $ to $S_{n,e}$ is an injection and the
restriction of $\phi $ to $S_{n,e+1}$ is a surjection. Choose a set $K,$ $%
\left\vert K\right\vert =k,$ of words in $%
\mathbb{Z}
^{n}$ so that the restriction of $\phi $ to $K$ is a bijection and $%
S_{n,e}\subseteq K\subset S_{n,e+1}.$ By Theorem \ref{6}, $\phi $ induces a
lattice tiling $\mathcal{T}$ of $%
\mathbb{Z}
^{n}$ by $K.$ Consider the lattice $\mathcal{L=\ker }(\phi ).$ For any two
words \thinspace $u,v\in \mathcal{L}$ we have $\rho _{L}(u,v)\geq 2e+1$ as $%
S_{n,e}\subseteq K.$ Since $\mathcal{T}$ is a tiling, to each word $w\in 
\mathbb{Z}
^{n}\,$there exists a copy $K_{w}$ of $K$ so that $w\in K_{w}$, and $%
K\subset S_{n,e+1}$ guaranties that $w$ is at distance $\leq e+1$ from at
least one word in $\mathcal{L}$. Thus, $\mathcal{L}$ constitutes a linear $%
QPL(n,e)$ code\medskip

\noindent Now, assume that $\mathcal{L\subset }%
\mathbb{Z}
^{n}$ is a linear $QPL(n,e)$ code. We will prove that there is a number $%
k,\left\vert S_{n,e}\right\vert \leq k<\left\vert S_{n,e+1}\right\vert ,$ so
that $k$ has an optimal embedding in $%
\mathbb{Z}
^{n}.$ It is well known \ that $%
\mathbb{Z}
^{n}/\mathcal{L}\simeq G,$ where $G$ is an abelian group. We show that $G$
has an optimal embedding, and $\left\vert S_{n,e}\right\vert \leq \left\vert
G\right\vert <\left\vert S_{n,e+1}\right\vert .$ Consider the natural
homomorphism $\phi :%
\mathbb{Z}
^{n}\rightarrow G$. Then $\ker (\phi )=\mathcal{L}$. Assume that there are
two words $u,v\in S_{n,e},u\neq v,$ such that $\phi (u)=\phi (v).$ Set $%
w=u-v.$ Then $\phi (w)=\phi (u-w)=\phi (u)-\phi (v)=0.$ Thus, $w\in \mathcal{%
L}$ is a codeword. However, this is a contradiction as $\rho _{L}(w,O)=\rho
_{L}(u,v)<2e+1$ which contradicts the condition (i) in the definition of $%
QPL(n,e)$ code. Hence we proved that $\phi $ is an injection on $S_{n,e},$
which at the same time implies that $\left\vert G\right\vert \geq \left\vert
S_{n,e}\right\vert .$\medskip

\noindent To prove that the restriction of $\phi $ to $S_{n,e+1}$ is
surjective, consider an element $g\in G$. Let $u$ be a word in $%
\mathbb{Z}
^{n}$\noindent\ with $\phi (u)=g.$ As $\mathcal{L}$ is a $QPL(n,e)$ code,
for each word $u\in 
\mathbb{Z}
^{n},$ there is a codeword $w\in \mathcal{L}$ so that $\rho _{L}(u,w)\leq
e+1,$ and for $u-w\,\ $\ we have $\rho _{L}(O,u-w)\leq e+1,$ i.e., $u-w\in
S_{n,e+1}.$ As $\phi $ is a homomorphism, $\phi (u-w)=\phi (u)-\phi
(w)=g-0=g.$\medskip
\end{proof}

\noindent In practical applications we deal with finite perfect Lee codes
over the alphabet $%
\mathbb{Z}
_{p}^{n}.\,\ $These codes are usually denoted as $PL(n,2,p)$ codes ($%
QPL(n,e,p)$ codes). As a corollary of Theorem \ref{15} we get:

\begin{corollary}
There is a linear $QPL(2,e,k)$ code for each $\left\vert S_{2,e}\right\vert
\leq k<\left\vert S_{2,e+1}\right\vert .$
\end{corollary}

\begin{proof}
By Theorem \ref{11}, each $k\geq 1$ has an optimal embedding in $%
\mathbb{Z}
^{2},$ and by Theorem \ref{15} there is a linear $QPL(n,e)$ code $\mathcal{L}
$, where $%
\mathbb{Z}
^{n}/\mathcal{L}$ $\simeq G$ is an abelian group of order $\left\vert
G\right\vert =k$. Denote by $\phi $ the natural homomorphism $\phi :%
\mathbb{Z}
^{n}\rightarrow G.$ Further, for the smallest period $p$ of $\mathcal{L}$ $\ 
$we have $p=l.c.m.\{ord(\phi (e_{i})),i=1,..,n\},$ where $ord(g)$ stands for
the order of the element $g$ in the group $G.$ Thus $p$ divides $\left\vert
G\right\vert ,$ hence $\mathcal{L}$ is a linear $QPL(n,e)$ code that is $k$%
-periodic, and thus induces a linear $QPL(n,e,k)$ code$.$\medskip
\end{proof}

\noindent By Theorem \ref{9} and Theorem \ref{15} we immediately get:

\begin{corollary}
For each $n>2,$ there are at most finitely many values of $e$ for which
there exists a linear $QPL(n,e)$ code.
\end{corollary}

\noindent Now we concentrate on the case of $n=3$ that is most likely to be
used in a real-life application. It follows from the result of Gravier et
al. \cite{Gra}, that there is no $PL(3,e)$ code for $e>1,$ so there is no
optimal embedding for $k=\left\vert S_{3,e}\right\vert $ in $%
\mathbb{Z}
^{3}.$\medskip

\noindent Set $K=[1,21]\cup \lbrack 27,50]\cup \{55\}\cup \lbrack
70,102]\cup \{117,145\}\cup \lbrack 147,151]\cup \lbrack 153,156]\cup 
\newline
[158,165]\cup \lbrack 167,172]\cup \lbrack 174,177]\cup
\{182,183,190,260,261,263,264,266,267,268,\newline
270\}\cup \lbrack 272,276]\cup
\{279,282,286,288,292,300,421,422,426,438,455\}.$

\begin{theorem}
If $k\in K,$ then $k$ has an optimal embedding in $%
\mathbb{Z}
^{3}.$ In particular, there is a linear $QPL(3,e)$ code for each $e,1\leq
e\leq 6.$
\end{theorem}

\begin{proof}
It suffices to prove that the cyclic group $%
\mathbb{Z}
_{k}$ has an optimal embedding in $%
\mathbb{Z}
^{3}$ for each $k\in K$. A required homomorphism $\phi ,$ uniquely
determined by the values of $\phi (e_{i}),i=1,2,3,$ has been found by a
computer search. For example, for $k=7,...,13,$ it suffices to choose $\phi
(e_{i})=i,i=1,2,3.$ The values $\phi (e_{i}),i=2,3,$ for the other $k\in K$
are given in Appendix, while $\phi (e_{1})=1$ except for $k=438$ where $\phi
(e_{1})=2.$\medskip

\noindent For $n=3,$ it is $\left\vert S_{3,e}\right\vert =\frac{4}{3}%
e^{3}+2e^{2}+\frac{8}{3}e+1.$ Thus, for $e=1,2,3,4,5,$ $6,7,$ we get that $%
\left\vert S_{3,e}\right\vert =7,25,63,129,231,377,575,$ respectively. To
prove the second part of the statement it suffices to notice that for each $%
e,1\leq e\leq 6,$ there is a $k\in K$ with $\left\vert S_{3,e}\right\vert
\leq k<\left\vert S_{3,e+1}\right\vert $.\medskip
\end{proof}

\noindent The last theorem asserts that $QPL(3,e)$ codes exist only for
finitely many values of $e.$ We note that the statement of the theorem could
be proved with the condition \textit{linear} dropped.

\begin{theorem}
There is no linear $QPL(3,e)$ code for $e\geq 55.$
\end{theorem}

\begin{proof}
Suppose that there is a linear $QPL(3,e)$ code. By Theorem \ref{15}, there
is a $k,\left\vert S_{3,e}\right\vert \leq k<\left\vert S_{3,e+1}\right\vert 
$, so that $k$ has an optimal embedding in $%
\mathbb{Z}
^{3}.$ Using the language of the proof of Theorem \ref{9}, this implies that
there is a lattice tiling of $%
\mathbb{R}
^{n}$ by a cluster $T_{k}$ of unit cubes, $L_{3,e}\subseteq T_{k}\subset
L_{n,e+1},$ with its volume $V(T_{k})=k.$ \ Also, by proof of Theorem \ref{9}%
, there is no tiling $\mathcal{T}{\normalsize \ }$of $%
\mathbb{R}
^{n}$ by a cluster of unit cubes $T_{k}$ with 
\begin{equation}
\frac{V(P_{3,e})}{V(T_{k})}>\alpha ,  \label{14}
\end{equation}%
where $\alpha $ is the packing efficiency of the regular polytope $P_{3,e}.$
The packing efficiency $\alpha =\frac{18}{19}$ of the regular octahedron has
been determined by Minkowski in \cite{M}. To prove the non-existence of a
linear $QPL(n,e)$ code we need to show that (\ref{14}) is satisfied by all $%
k,\left\vert S_{3,e}\right\vert \leq k<\left\vert S_{3,e+1}\right\vert .$
Clearly, it suffices to show that (\ref{14}) is satisfied by $k=\left\vert
S_{3,e+1}\right\vert -1$ as $\frac{V(P_{3,e})}{V(T_{k+1})}<\frac{V(P_{3,e})}{%
V(T_{k})}.$ Solving (\ref{14}) for $V(P_{3,e})=\frac{(2e+1)^{3}}{3!}$ and $%
V(T_{k})=\left\vert S_{3,e+1}\right\vert -1=\frac{4}{3}(e+1)^{3}+2(e+1)^{2}+%
\frac{8}{3}(e+1)$ we get that there is no linear $QPL(3,e)$ code for $e\geq
55.$
\end{proof}

\section{Appendix- Optimal Embeddings in $%
\mathbb{Z}
^{3}$}

\hspace*{-3cm} 
\begin{tabular}{|l|l|l|l|l|l|l|l|l|l|l|l|l|l|l|}
\hline
\textbf{k} & $\phi (e_{2})$ & $\phi (e_{3})$ & \textbf{k} & $\phi (e_{2})$ & 
$\phi (e_{3})$ & \textbf{k} & $\phi (e_{2})$ & $\phi (e_{1})$ & \textbf{k} & 
$\phi (e_{2})$ & $\phi (e_{3})$ & \textbf{k} & $\phi (e_{2})$ & $\phi
(e_{3}) $ \\ \hline
\textbf{14} & 2 & 5 & \textbf{48} & 7 & 18 & \textbf{95} & 6 & 37 & \textbf{%
169} & 10 & 72 & \textbf{300} & 14 & 132 \\ \hline
\textbf{15} & 2 & 4 & \textbf{49} & 7 & 11 & \textbf{96} & 6 & 37 & \textbf{%
170} & 9 & 64 & \textbf{421} & 16 & 182 \\ \hline
\textbf{16} & 2 & 6 & \textbf{50} & 8 & 12 & \textbf{97} & 7 & 36 & \textbf{%
171} & 12 & 70 & \textbf{422} & 72 & 112 \\ \hline
\textbf{17} & 2 & 6 & \textbf{55} & 5 & 21 & \textbf{98} & 7 & 36 & \textbf{%
172} & 11 & 52 & \textbf{426} & 36 & 50 \\ \hline
\textbf{18} & 2 & 7 & \textbf{70} & 16 & 25 & \textbf{99} & 7 & 37 & \textbf{%
174} & 14 & 34 & \textbf{438} & 45 & 122 \\ \hline
\textbf{19} & 2 & 7 & \textbf{71} & 7 & 30 & \textbf{100} & 6 & 22 & \textbf{%
175} & 10 & 53 & \textbf{455} & 16 & 199 \\ \hline
\textbf{20} & 5 & 8 & \textbf{72} & 8 & 30 & \textbf{101} & 11 & 27 & 
\textbf{176} & 16 & 41 &  &  &  \\ \hline
\textbf{21} & 2 & 8 & \textbf{73} & 6 & 21 & \textbf{102} & 10 & 43 & 
\textbf{177} & 17 & 28 &  &  &  \\ \hline
\textbf{27} & 5 & 8 & \textbf{74} & 8 & 20 & \textbf{117} & 16 & 22 & 
\textbf{182} & 35 & 64 &  &  &  \\ \hline
\textbf{28} & 5 & 8 & \textbf{75} & 6 & 22 & \textbf{145} & 9 & 61 & \textbf{%
183} & 21 & 29 &  &  &  \\ \hline
\textbf{29} & 5 & 13 & \textbf{76} & 7 & 18 & \textbf{147} & 9 & 62 & 
\textbf{190} & 22 & 30 &  &  &  \\ \hline
\textbf{30} & 5 & 8 & \textbf{77} & 7 & 18 & \textbf{148} & 32 & 46 & 
\textbf{260} & 40 & 94 &  &  &  \\ \hline
\textbf{31} & 5 & 8 & \textbf{78} & 7 & 30 & \textbf{149} & 12 & 52 & 
\textbf{261} & 36 & 61 &  &  &  \\ \hline
\textbf{32} & 6 & 9 & \textbf{79} & 6 & 32 & \textbf{150} & 16 & 26 & 
\textbf{263} & 11 & 97 &  &  &  \\ \hline
\textbf{33} & 5 & 8 & \textbf{80} & 6 & 21 & \textbf{151} & 10 & 63 & 
\textbf{264} & 16 & 55 &  &  &  \\ \hline
\textbf{34} & 5 & 8 & \textbf{81} & 8 & 21 & \textbf{153} & 17 & 41 & 
\textbf{266} & 40 & 127 &  &  &  \\ \hline
\textbf{35} & 5 & 8 & \textbf{82} & 7 & 26 & \textbf{154} & 8 & 58 & \textbf{%
267} & 12 & 99 &  &  &  \\ \hline
\textbf{36} & 5 & 8 & \textbf{83} & 6 & 31 & \textbf{155} & 9 & 66 & \textbf{%
268} & 40 & 98 &  &  &  \\ \hline
\textbf{37} & 5 & 8 & \textbf{84} & 6 & 31 & \textbf{156} & 10 & 47 & 
\textbf{270} & 14 & 117 &  &  &  \\ \hline
\textbf{38} & 6 & 9 & \textbf{85} & 7 & 25 & \textbf{158} & 9 & 48 & \textbf{%
272} & 14 & 118 &  &  &  \\ \hline
\textbf{39} & 6 & 9 & \textbf{86} & 6 & 32 & \textbf{159} & 10 & 67 & 
\textbf{273} & 12 & 81 &  &  &  \\ \hline
\textbf{40} & 4 & 15 & \textbf{87} & 6 & 32 & \textbf{160} & 14 & 34 & 
\textbf{274} & 102 & 128 &  &  &  \\ \hline
\textbf{41} & 4 & 10 & \textbf{88} & 6 & 26 & \textbf{161} & 10 & 68 & 
\textbf{275} & 44 & 60 &  &  &  \\ \hline
\textbf{42} & 6 & 10 & \textbf{89} & 6 & 37 & \textbf{162} & 34 & 75 & 
\textbf{276} & 104 & 117 &  &  &  \\ \hline
\textbf{43} & 6 & 10 & \textbf{90} & 6 & 37 & \textbf{163} & 11 & 68 & 
\textbf{279} & 54 & 89 &  &  &  \\ \hline
\textbf{44} & 6 & 10 & \textbf{91} & 7 & 24 & \textbf{164} & 10 & 69 & 
\textbf{282} & 74 & 100 &  &  &  \\ \hline
\textbf{45} & 6 & 10 & \textbf{92} & 10 & 38 & \textbf{165} & 9 & 71 & 
\textbf{286} & 14 & 88 &  &  &  \\ \hline
\textbf{46} & 6 & 21 & \textbf{93} & 6 & 26 & \textbf{167} & 15 & 39 & 
\textbf{288} & 84 & 106 &  &  &  \\ \hline
\textbf{47} & 6 & 19 & \textbf{94} & 6 & 26 & \textbf{168} & 12 & 69 & 
\textbf{292} & 40 & 102 &  &  &  \\ \hline
\end{tabular}

\end{document}